\documentclass[11pt]{article}
\usepackage{dsfont}
\usepackage{amsmath}
\usepackage{amsthm}
\usepackage{amssymb}
\usepackage{algorithmicx}
\usepackage{algorithm}
\usepackage{algpseudocode}
\usepackage{graphicx}
\usepackage{url}
\usepackage{hyperref}
\usepackage{mathabx}
\usepackage[letterpaper,margin=1in]{geometry}
\usepackage{authblk}

\newcommand{\set}[1]{\{#1\}}

\renewcommand{\setminus}{-}
\newcommand{\R}{\mathbb{R}}    
\theoremstyle{plain}
\newtheorem{theorem}{Theorem}[section]
\newtheorem{lem}[theorem]{Lemma}
\newtheorem{definition}[theorem]{Definition}

\newtheorem{prop}[theorem]{Proposition}
\newtheorem{cor}{Corollary}
\newtheorem{fact}{Fact}

\theoremstyle{definition}
\newtheorem{defn}{Definition}[section]

\theoremstyle{remark}

\newtheorem*{note}{Note}

\newcommand{\eps}{\varepsilon}
\newcommand{\opt}{\text{OPT}}
\newcommand{\vor}{\text{Vor}}
\newcommand{\dist}{\text{dist}}

\newcommand{\cost}{\text{cost}}

\newcommand{\calC}{\mathcal{C}}

\newcommand{\calF}{\mathcal{F}}
\newcommand{\calL}{\mathcal{L}}
\newcommand{\calG}{\mathcal{G}}

\newcommand{\calM}{\mathcal{M}}
\newcommand{\calR}{\mathcal{R}}

\newcommand{\local}{\mathcal{L}}

\newcommand{\globalS}{\mathcal{G}}
\newcommand{\globalSs}{\mathcal{G}^{*}}
\newcommand{\Vloc}{V_{\local}}
\newcommand{\Vglob}{V_{\globalS}}
\newcommand{\globtoloc}{\text{Reassign}_{\globalSs \mapsto \local}}
\newcommand{\loctoglob}{\text{Reassign}_{\local \mapsto \globalSs}}

\newcommand{\costPeps}{(1+\eps_1)^p}

\newcommand{\costPinv}{\eps_1^{-p}}

\newif\ifshort
  \shortfalse 

\title{Local search yields approximation schemes for k-means and k-median in Euclidean and minor-free metrics}
\author[1]{Vincent Cohen-Addad \thanks{Research funded by the French ANR Blanc project ANR-12-BS02-005 (RDAM)}}
\author[2]{Philip N. Klein \thanks{Research funded by NSF Grant 
    CCF-10-12254
    with additional support from the Radcliffe Institute of Advanced
    Study, Harvard University}}
\author[3]{Claire Mathieu \thanks{Research funded by the French ANR Blanc
  project ANR-12-BS02-005 (RDAM)}}

\affil[1]{D\'epartement d'informatique, \'Ecole normale sup\'erieure, 
  France}
\affil[2]{Brown University, United States}
\affil[3]{CNRS\\D\'epartement d'informatique, 
  \'Ecole normale sup\'erieure,France}

\date{}

\begin{document}
\maketitle

\begin{abstract}
  We give the first polynomial-time approximation schemes (PTASs) for the following problems: (1) uniform facility location in edge-weighted planar graphs; (2) $k$-median and $k$-means in edge-weighted planar graphs; (3) $k$-means in Euclidean space of bounded dimension. Our first and second results extend to minor-closed families of graphs. 
  All our results extend to cost functions that are the $p$-th power of the shortest-path distance. The algorithm is local search where the local neighborhood of a solution $S$ consists of all solutions obtained from $S$ by removing and adding $1/\eps^{O(1)}$ centers.
\end{abstract}
\thispagestyle{empty}
\newpage
\setcounter{page}{1}
\section{Introduction}
In this paper, we address three fundamental problems, facility location,
$k$-median and $k$-means clustering, in two settings, 
graphs and Euclidean spaces.  
The problem of approximating $k$-means clustering in
low-dimensional Euclidean space has been studied since at least
1994~\cite{IKI94};
since then, many researchers have given approximation
schemes that are polynomial for fixed $k$ but exponential in $k$.  
Very recently, building on~\cite{CM15}, a {\em bicriteria}
polynomial-time approximation scheme has been given~\cite{BaV15} for $k$-means: it finds
$(1+\epsilon)k$ centers whose cost is at most $1+\epsilon$ times the
cost of an optimal $k$-means solution.  As the authors of
that paper point out, it remained an open question whether there is a true
polynomial-time approximation scheme for $k$-means in the plane (where
$k$ is considered part of the input); the
best polynomial-time approximation bound known was $9+\epsilon$.  In
this paper, we resolve this open question by giving the first
polynomial-time approximation scheme for arbitrary (i.e. nonconstant)
$k$ in low-dimensional Euclidean space.

Our analysis of the $k$-means approximation scheme shows that it can
also be applied to graphs belonging to a fixed nontrivial minor-closed family.\footnote{Contracting an edge of a graph means
  identifying its endpoints and then removing it. A graph $H$ is a
  {\em minor} of graph $G$ if $H$ can be obtained from $H$ by edge
  deletions and edge contractions.  The family of planar graphs, for
  example, is closed under taking minors, as is the family of graphs
  embeddable on a surface of genus $g$, for any fixed integer $g>0$.
  We say a minor-closed family is nontrivial if it omits at least one
  graph.}  

For example, for any fixed integer $g$, 
graphs embeddable on a surface of genus $g$ form such a family.  In
particular, planar graphs forms such a family.  
Thus we also obtain the first
polynomial-time approximation scheme for $k$-means in planar graphs.

The problems of (uncapacitated) metric facility location and $k$-median in graphs has similarly been studied
for many years. The first polynomial-time approximation algorithm, with a logarithmic
performance guarantee, was given by Hochbaum in 1982~\cite{Hochbaum82}.  The first
polynomial-time approximation algorithm to achieve a constant
approximation ratio was given by Shmoys et al.
~\cite{STA1997} in 1997. It was later improved by Jain and Vazirani
\cite{JaV01} and by Arya et al.~\cite{AGKMMP04}.  The current best approximation algorithm for metric (uncapacitated) facility location, due to Li, has approximation ratio 1.488~\cite{li2013}.
Guha and Khuller ~\cite{GuK99} proved that there exists no polynomial-time
approximation algorithm with approximation ratio of 1.463 for metric
facility location problem unless
$NP\subseteq DTIME[n^{O(\log\log n)}]$. The current best approximation 
ratio for  the $k$-median problem is $1+\sqrt{3}$ by Li and Svensson~\cite{LiS13}.

In order to obtain a substantially better approximation ratio,
therefore, one must restrict attention to special metrics.  Because
facility location problems often arise on the surface of the earth, it
is natural to consider the metrics induced by planar graphs.  Researchers have
been trying to find a polynomial-time approximation scheme for the
planar restriction for many years.  An unpublished
manuscript~\cite{Ageev} by Ageev dating back at least to 2001 addressed the
planar case via a straightforward application of Baker's
method~\cite{Baker94}, giving an algorithm whose performance on an
instance depends on how much of the cost of the optimal solution is
facility-opening cost.  Despite the title of the manuscript, the
algorithm is {\em not} an approximation scheme for instances with
arbitrary weights.  Since then there have been no results on the
problem despite efforts by several researchers in the area.

In this paper, we give the first polynomial-time approximation scheme
for (uncapacitated uniform) facility location and $k$-median where the metric is 
that induced by a planar graph or, more generally a graph belonging to a fixed nontrivial
minor-closed family.


\subsection{Results}

We describe a simple and natural, and previously studied local-search algorithm for clustering
problems, parameterized by the desired cluster size $k$, the objective
function $\cost(\cdot)$, and a parameter $s$ governing the
local-search neighborhood.

\begin{algorithm}
  \caption{Local Search for finding $k$ clusters}
  \label{algo:Clustering}
  \begin{algorithmic}[1]
    \State \textbf{Input:} A metric space and associated cost function
    $\cost(\cdot)$, an $n$-element set $C$
    of points, error parameter $\eps>0$, positive integer parameters $k$ and $s$ 
    \State $S \gets$ Arbitrary size-$k$ set of points
    \While{$\exists$ $S'$ s.t. $|S'|\leq k$  \textbf{and} $|S \setminus S'| + |S' \setminus S| \leq s$
      \textbf{and} cost($S'$) $\le$ $(1-\eps/n) \cost(S)$\\}
    \State $S \gets S'$
    \EndWhile
    \State \textbf{Output:} $S$     
  \end{algorithmic}
\end{algorithm}

We consider two kinds of metric spaces.  For any fixed positive
integer $d$, we consider $\R^d$ equipped with Euclidean distance.  For
any undirected edge-weighted graph $G$, we consider the {\em metric
  completion} of $G$, i.e. the metric space whose points are the
vertices of $G$ and where the distance between $u$ and $v$ is defined
to be the length of the shortest $u$-to-$v$ path in $G$ with respect
to the given edge-weights.





\begin{theorem}[Euclidean Spaces] \label{thm:ClusteringEuclid} For any
  fixed integers $p,d>0$, there is a constant $c$ such that, for any $0<\eps<1/2$,
  applying Algorithm~\ref{algo:Clustering} to the $d$-dimensional
  Euclidean space with cost function
  $$\cost(S) = \sum_{c \in C} ( \min_{u \in S} \dist(c,u))^p$$
  and  $s=1/\eps^{c}$ yields a solution $S$ whose cost 
  is at most $1+\eps$ times the minimum.
\end{theorem}
When $p=2$, the objective function is the {\em
  $k$-means} objective function.  When $p=1$,  the objective function
is that of {\em $k$-median}.

When the metric space is $\R^d$, it is not trivial to implement an
iteration of Algorithm~\ref{algo:Clustering}.  However, as observed
in~\cite{BaV15} (see~\cite{IKI94}), there is a method using an arrangement of
algebraic surfaces to execute an iteration in $n^{O(ds)}$ time.  The
number of iterations is $O(n/\eps)$ (see \cite{AGKMMP04,CM15}).
The running time is therefore polynomial for fixed $p,d,\eps$.
We obtain the following.

\begin{cor} For any integer $d>0$, there is a polynomial-time approximation scheme for
  the $k$-means problem in $d$-dimensional Euclidean spaces.
\end{cor}

Algorithm~\ref{algo:Clustering} can also be applied to the metric completion of a graph.

\begin{theorem}[Graphs] \label{thm:Clustering}
  Let $\mathcal K$ be a nontrivial minor-closed family of
  edge-weighted graphs. For any fixed integer $p>0$, 
  there is a constant $c$ with the
  following property.  For any $0<\eps<1/2$, for
  any $G\in \mathcal K$, Algorithm~\ref{algo:Clustering} applied to
  the metric completion of $G$ with cost function
  $$\cost(S)=\sum_{c \in C} ( \min_{u \in S} \dist(c,u))^p$$
  and with $s=1/\eps^{c}$ yields a solution $S$ whose cost
  is at most $1+\eps$ times the minimum.
\end{theorem}

It is straightforward to implement Algorithm~\ref{algo:Clustering}
applied to the metric completion of a graph.  As before, the number of
iterations is $O(n/\eps)$ where $n$ is the number of clients.
We therefore obtain:

\begin{cor} There is a polynomial-time approximation scheme for
  $k$-means and for $k$-median in planar graphs and in bounded-genus graphs.

  More generally, for any nontrivial minor-closed family of
  edge-weighted graphs, there is a polynomial-time approximation
  schemes for $k$-means and for $k$-median for graphs in that family.
\end{cor}

The local-search algorithm is easily adapted to the case where we do
not specify the number of clusters but instead specify a per-cluster
cost.  This case includes a variant of the \emph{facility location} problem.

\begin{defn}[Uncapacitated Uniform Facility Location]
The \emph{Uncapacitated Uniform Facility Location} problem is as follows: given a
finite metric space, a subset $C$ of points, and a facility opening cost $f$,
find a subset $S$ of points that minimizes $\cost(S)  =f|S| + \sum_{c \in C}
\min_{u \in S} \dist(c,u)$.
\end{defn}


To address this problem, we use a simple modification of the
local-search algorithm given earlier.  

\begin{algorithm}
  \caption{Local Search for Uniform Facility Location}
  \label{algo:FacilityLocation}
  \begin{algorithmic}[1]
    \State \textbf{Input:} A metric space and associated cost function
    $\cost(\cdot)$, an $n$-element set $C$
    of points, error parameter $\eps>0$, facility opening cost
    $f>0$, positive integer parameter $s$
    \State $S \gets$ Arbitrary subset of $\calF$.
    \While{$\exists$ $S'$ s.t. $|S \setminus S'| + |S' \setminus S| \leq s$
      \textbf{and} cost($S'$) $\le$ $(1-\eps/n)$ cost($S$)\\}
    \State $S \gets S'$
    \EndWhile
    \State \textbf{Output:} $S$     
  \end{algorithmic}
\end{algorithm}

\begin{theorem}\label{thm:FC} Fix a nontrivial minor-closed family
  $\mathcal K$ of
  graphs.  There is a constant $c$ such that, when Algorithm~\ref{algo:FacilityLocation} is applied to
  the metric completion of a graph in $\mathcal K$ with
$$\cost(S)=|S|f + \sum_{c \in C} ( \min_{u \in S} \dist(c,u))^p$$
and  $s=1/\eps^c$, the output has cost at most $1+\eps$ times the
minimum.
\end{theorem}

In fact, for $p=1$, setting $s=c/\eps^2$ suffices to achieve a
$1+\eps$ approximation.  The theorem implies the following:

\begin{cor} \label{cor:FC}
Fix a nontrivial minor-closed family
  $\mathcal K$ of edge-weighted
  graphs.  There is a polynomial-time approximation scheme for uniform
  uncapacitated facility location in graphs of $\mathcal K$.
\end{cor}


\subsection{Related work}
In arbitrary metric spaces, it is NP-hard to approximate  the $k$-median and $k$-means problems 
within a factor of $1+2/e$ and $1+3/e$ respectively, see Guha and Khuller~\cite{GuK99} and 
Jain et al.~\cite{JMS02}. 
In the case of Euclidean space, Guruswami and Indyk~\cite{GI03} showed that there is no PTAS for 
$k$-median if both $k$ and $d$ are part of the input.
More recently, Awasthi et al.~\cite{ACKS15} showed APX-Hardness for $k$-means if both $k$ and $d$ are part of the input. 

In Euclidean spaces, $(1+\varepsilon)$-approximation algorithms for $k$-median have been
proposed when $k$ or $d$ is fixed.
For example, when $k$ is fixed, there exists different PTAS (See \cite{BHPI02,KSS04,KSS10,FMS07,HaM04} and ~\cite{FeL11} for the best
known so far).
When $d$ is fixed, Arora et al. gave the first PTAS~\cite{ARR98} for the $k$-median problem.
This result was subsequently improved to an efficient PTAS by Kolliopoulos et al.~\cite{KoR07} 
and Har-Peled et al.~\cite{HaK07,HaM04}.

For the $k$-means problem, Kanungo et al.~\cite{KMNPSW04} showed that
local search achieves a $9+\eps$-approximation in general metrics
and this remains the best known approximation guarantee so far even for 
fixed $d$.
There are also a variety of results for $k$-means and $k$-median 
when the input has some stability conditions 
(see for example \cite{AwS12,ABS10,BaL16,BBG09,BiL12,KuK10,ORSS12}) 
or in the context of smoothed analysis (see for example \cite{ArV09,AMR11}). 

Local Search for metric $k$-median was first analyzed by Korupolu et al~\cite{KPR00}. 
They gave a bicriteria approximation using $k\cdot(1+\varepsilon)$ centers an achieving a cost of at most 
$3+5/\varepsilon$ times the cost of the optimum $k$-clustering. 
This was later improved to $k\cdot(1+\varepsilon)$ centers an achieving a cost 
of at most $2+2/\varepsilon$ times the cost of the optimum $k$-clustering by Charikar an Guha~\cite{ChG05}.
Arya et al.~\cite{AGKMMP04} gave the first analysis showing that Local Search with a neighborhood of 
size $1/\varepsilon$ gives a $3+2\varepsilon$ approximation to $k$-median.
Moreover, they show that this bound is tight.
As mentioned earlier, Kanungo et al.~\cite{KMNPSW04} showed that local
search is a $9+\eps$-approximation for $k$-means in general metrics.
Local search is a very popular algorithm for clustering and has been 
widely used : see~\cite{BeT10}  in the context of parallel algorithms,
~\cite{GMMMO03} in the streaming model and ~\cite{BBLM14} for 
distributed computing. See \cite{AaL97} for a general introduction
to theory and practice of local search.

\paragraph{Note added:} 
After we had written up our results and while we were editing the submission, we noticed a recent ArXiv paper \cite{FRS16} that 
has similar results for doubling metrics.

\section{Techniques}

\subsection{$r$-divisions in minors}

One key ingredient in our analyses is the existence of a certain kind
of decomposition of the input called a \emph{weak $r$-division}.  The 
concept (in a
stronger form) is due to Frederickson~\cite{Frederickson87} in the
context of planar graphs.  It is straightforward to extend it to any
family of graphs with balanced separators of size
sublinear-polynomial.  We also define a weak $r$-division for points
in a Euclidean space, and show that such a decomposition always
exists.  These definitions and results are in
Sections~\ref{sec:graph-r-division} and~\ref{sec:Euclidean-r-division}.
Note that $r$-divisions play no role in our algorithm; only the
analysis uses them.

Chan and Har-Peled~\cite{CH12} showed
that local search can be used to obtain a PTAS for (unweighted)
maximum independent pseudo-disks in the plane, which implies the
analogous result for planar graphs.  More generally, Har-Peled and
Quanrud~\cite{HQ15} show that local search can be used
to obtain PTASs for several problems including {\em independent set},
{\em set cover}, and {\em dominating set}, in graphs with polynomial
expansion.  These graphs have small separators and therefore
$r$-divisions.  However, our analysis of local search for clustering
requires not only that the {\em input graph} have an $r$-division but
that a {\em minor} of the input graph have an $r$-division.  This is
not true of graphs of polynomial expansion.  Indeed, we show in
Section~\ref{sec:tightness} that there are low-density graphs in
low-dimensional space (which are therefore polynomial-expansion
graphs) for which our local-search algorithm produces a solution that
is worse than the optimum by at least a constant factor.  

Thus one of our technical contributions is showing how to take
advantage of a property possessed by nontrivial minor-closed graph
families that is not possessed by polynomial-expansion graph families.

\subsection{Isolation}
\label{sec:isolation}
In order to obtain our approximation schemes for $k$-means and
$k$-median clustering, we need another technique. As mentioned
earlier, a bicriteria approximation scheme for $k$-means was already
known; the solution it returns has more than $k$ centers.  It seems
hard to avoid an increase in the number of centers in comparing a
locally optimal solution to a globally optimal solution.  It would help if
we could show that the globally optimal solution could be modified so
as to {\em reduce}
the number of centers
below $k$ while only slightly increasing the cost; we could then
compare the local solution to this modified global solution, and 
the increase in the number of centers would leave the number no more
than $k$.

Unfortunately, we cannot unconditionally reduce the number of
centers.  However, consider a globally optimal solution $\globalS$ and
a locally optimal solution $\local$.  A facility $f$ in $\globalS$
might correspond to a facility $\ell$ in $\local$ in the sense that
they serve almost exactly the same set of clients.  In this case, we
say the pair $(f, \ell)$ is {\em 1-1 isolated} (the formal definition is
below).  Such centers do not contribute much to the increase in cost
in going from global solution to local solution, so let's ignore
them.  Among the remaining centers of $\globalS$, there are a
substantial number that can be removed without the cost increasing much.
The analysis of the local solution then proceeds as discussed above.

We now give the formal definition of 1-1 isolated.

\begin{definition}\label{def:1-1-isolated}
  Let $\eps <1/2$ be a positive constant and $\local$ and $\globalS$ be two solutions for the $k$-clustering problem with parameter $p$.
  Given a facility $f_0 \in \globalS$ and a facility $\ell\in \local$, 
we say that the pair
  $(f,\ell)$ is  {\em 1-1-isolated} if 
  most of the clients served by $\ell$ in $\local$ are served by $f$ in $\globalS$, and
  most of the clients served by $f$ in $\globalS$ are served by $\ell$ in $\local$ : in other words, 
  $$|\Vloc(\ell) \cap \Vglob(f)  |   \ge \left\{ \begin{array}{l} (1-\eps)   |\Vloc(\ell)|    \\  (1-\eps) |\Vglob(f)| \end{array}\right. $$
\end{definition}

\begin{theorem}\label{thm:deletion}
  Let $\eps <1/2$ be a positive constant and $\local$ and $\globalS$ be two solutions for the $k$-clustering problem with exponent $p$.
  Let $\bar k$ denote the number of facilities $f$ of $\globalS$  that
  are not in a 1-1 isolated region. 
  There exists a set $S_0$ of facilities of $\globalS$ of size at least $\eps^3 \bar k/6$ that can be removed 
  from $\globalS$ at low cost:  $\cost(\globalS \setminus S_0) \le (1+2^{3p+1}\eps) \cost(\globalS) +  2^{3p+1}\eps\cost(\local)$.
\end{theorem}

Note that the following theorem does not assume that $\local$ is a local optimum and $\globalS$ is an optimal solution. Thus 
we believe that this theorem can be of broader interest.
We now define the concept of \emph{isolated} regions; 1-1-isolated regions correspond to the special case of 
isolated regions when $\local_0 $ consists of a single facility.  
\begin{defn}[Isolated Region] \label{defn:isolated}
Given a facility $f_0 \in \globalS$ and a set of facilities $\local_0 \subseteq \local$, 
we say that the pair $(f_0, \local_0)$ is an {\em isolated region} if 
  \begin{itemize}
    \item For each facility $f' \in \local_0$, most of the clients served by $f'$ in $\local$ are served by $f_0$ in $\globalS$: in other words, 
    $| V_\local(f') \cap  V_\globalS(f_0) | \ge (1-\eps)|V_\local(f')|$, and
    \item Most of the clients served by $f_0$ in $\globalS$ are served by facilities of $\local_0$ in $\local$: in other words, $|V_\local(\local_0) \cap V_\globalS(f_0)| \ge (1-\eps)|V_\globalS(f_0)|$;
  \end{itemize}
\end{defn}

Finally, if $(f_0, \local_0)$ is an isolated region, we say that $f_0$
and the elements of $\local_0$ are {\em isolated}.

\subsection{Tightness of the results}
\label{sec:tightness}
\begin{prop}
  For any $w, t$, there exists an infinite family of graphs 
  excluding $K_{w}$ as a $t$-shallow minor such that for any constant
  $\eps$, local search with neighborhoods of size $1/\eps$
  might return a solution of cost at least $3 \opt$.
\end{prop}

See Figure \ref{fig:tightness} and \cite{AGKMMP04} for a complete proof that local search
performs badly on the instance depicted in the figure.

\begin{prop}
  For any $\rho$, there exists an infinite family of graphs 
  that are low-density graphs such that for any constant
  $\eps$, local search with neighborhoods of size $1/\eps$
  might return a solution of cost at least $3 \opt$.
\end{prop}

The proposition follows from encoding the graph of Figure 
\ref{fig:tightness} as an low-density graph.

We additionally remark that Awasthi et al.~\cite{ACKS15} show that the $k$-means problem is APX-Hard for 
any non-constant dimension $d$.
Moreover, Kanungo et al.~\cite{KMNPSW04} give an example where local search returns a solution
of cost at least $9\opt$.

\begin{figure}
  \begin{center}
    \includegraphics[scale = 0.7]{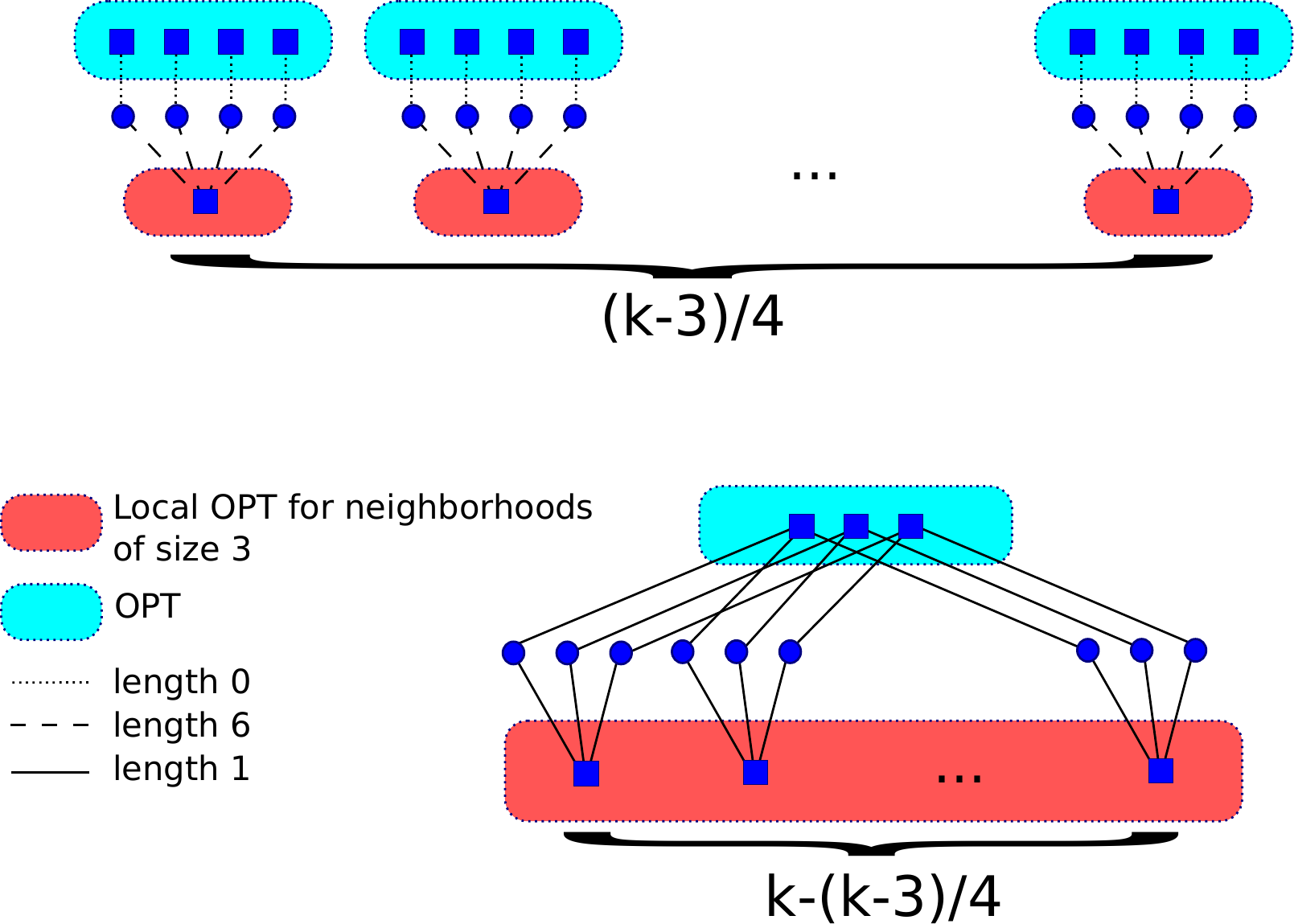}
  \end{center}
  \caption{
    This instance contains the complete
    bipartite graph $K_{3,3}$ as 
    a 1-shallow minor (and so is not planar) but 
    no non-trivial 0-shallow minor. Clients are denoted by circle and
    candidate centers by squares.
    There exists a local (for neighborhoods of size 3) optimum
    whose cost is at least $3\opt$. This example can be
    generalized to handle locality for neighborhoods of size $1/\eps$
    for any constant $\eps>0$.
    This is based on an example of 
    \cite{AGKMMP04} and can be extended to form a $i$-shallow minor graph for 
    any $i = o(n)$.
  }
  \label{fig:tightness}
\end{figure}

\section{Preliminaries}

We will use the following technical lemma in order to give a general proof that encompasses 
both the cases of $k$-median and $k$-means.
Throughout the paper we assume $p$ constant and define $\eps_1$ to be a positive constant. 
\begin{lem}\label{lem:TI}
  Let $p \ge 0$ and $0< \eps_1< 1/2$. For any $a,b,c \in A \cup F$, we have
  $$\dist(a,b)^p \le 
  \begin{cases}
    (1+\eps_1)^p (\dist(a,c)^p + \dist(c,b)^p/\eps_1^p) \\    
    2^p(\dist(a,c)^p + \dist(c,b)^p).
  \end{cases}
  $$ 
\end{lem}
\begin{proof}
  By the triangular inequality,
  $$ \dist(a,b)^p = (\dist(a,c) + \dist(c,b))^p \le 
  \begin{cases}
    (1+\eps_1)^p \dist(a,c)^p &\mbox{If $\dist(c,b) \le \eps_1 \dist(a,c)$.} \\
    (1+\eps_1)^p\dist(c,b)^p/\eps^p_1 &\mbox{Otherwise.} \\
  \end{cases}
  $$
  Moreover, by the binomial theorem,
  $ \dist(a,b)^p = (\dist(a,c) + \dist(c,b))^p \le 2^p (\cost(a,c)^p  + \cost(c,b)^p).$
\end{proof}

\subsection{Graph $r$-division and other definitions}
\label{sec:graph-r-division}
For a graph $G$, we use $V(G)$ and $E(G)$ to denote the set of
vertices of $G$ and the set of edges of $G$, respectively.
For a subgraph $H$ of $G$, the {\em vertex boundary} of $H$ in $G$, denoted $\partial_G(H)$,
is the set of vertices $v$ such that $v$ is in $H$ but has an incident
edge that is not in $H$.  (We might write $\partial(H)$ if $G$ is
unambiguous.)  A vertex in the vertex boundary of $H$ is called a {\em
  boundary vertex} of $H$.  A vertex of $H$ that is not a boundary
vertex of $H$ is called an {\em internal vertex}.  We denote the set
of internal vertices of $H$ as $\mathcal I(H)$.

\begin{definition}\label{def:rdivision} 
Let $c_1$ and $c_2$ be constants (depending on $\mathcal
  G$).  For a number $r$, a  {\em
weak $r$-division} of a graph $G$ (with respect to $c_1, c_2$) is a collection $\mathcal R$ of
subgraphs of $G$, called {\em regions}, with the following properties.
\begin{enumerate}
\item Each edge of $G$ is in exactly one region. \label{def:rdivisiona}
\item The number of regions is at most $c_1 |V(G)|/r$.\label{def:rdivisionb}
\item Each region contains at most $r$ vertices.\label{def:rdivisionc}
\item The number of boundary vertices, summed over all regions, is at
  most $c_2|V(G)|/r^{1/2}$.\label{def:rdivisiond}
\end{enumerate}
\end{definition}

A family of graphs $F$ is said to be closed under taking minor (\emph{minor-closed}) if for any graph $G \in F$,
for any minor $H$ of $G$, we have $H \in F$.

\begin{theorem}[Frederickson~\cite{Frederickson87} + Alon, Seymour,
  and Thomas~\cite{AlonST}] 
\label{thm:r-division}
  Let $\mathcal K$ be a nontrivial minor-closed family of graphs.  There exist $c_1,c_2$ such that
  every graph in $\mathcal K$ has a weak $r$-division with respect to
  $c_1, c_2$.
\end{theorem}
\ifshort
\else
\begin{proof} Alon, Seymour, and Thomas ~\cite{AlonST} proved a separator theorem for the family of 
graphs excluding a fixed graph as a minor.  Any nontrivial minor-closed family excludes some graph
 as a minor (else it is trivial).  Frederickson~\cite{Frederickson87}
 gave a construction for a stronger kind of $r$-division of a
 planar graph.   The construction
 uses nothing of planar graphs except that they have such separators.
\end{proof}
\fi

Let $G$ be an undirected graph with edge-lengths.  
Fix an arbitrary priority ordering of the vertex set $V(G)$. 
For every subset $S$ of $V(G)$, we define the {\em Voronoi partition with respect to $S$}.
For each vertex $v\in S$, the {\em Voronoi cell with center $v$},
denoted $V_S(v)$, is the set of vertices 
that are closer to $v$ than to any other vertex in $S$, breaking ties in favor of the highest-priority vertex of $S$.
\begin{fact}
  \label{fact:connected_subgraph}
  For any $S$, for any vertex $v \in S$, the induced subgraph $G[V_S(v)]$ is a connected subgraph of $G$.
\end{fact}
\ifshort
\else
\begin{proof}
Let $u\in V_S(v)$, and let $p$ denote a $v$-to-$u$ shortest path.  
Let $w$ be a vertex on $P$.  Assume for a contradiction that, for some vertex $v'\in S$, either the $v'$-to-$w$ shortest path $p'$ is shorter than the shortest $v$-to-$w$ path, or it is no longer and $v'$ has higher priority than $v$.  Replacing the $v$-to-$w$ subpath of $p$ with $p'$ yields a $v$-to-$u$ path that either is shorter than $p$ or is no longer than $p$ and originates at a higher-priority vertex than $v$.  
\end{proof}
\fi

It follows that, for any vertex $v$ of $G$, contracting the edges of
the subgraph $G[V_S(v)]$ yields a single vertex.

\begin{definition} \label{def:graph-Voronoi}
We define $G_{\vor(S)}$ as the graph obtained from $G$ by contracting every edge of $G[V_S(v)]$ for every vertex $v\in S$. 
For each vertex $v\in S$, we denote by $\hat v$ the vertex of
$G_{\vor(S)}$ resulting from contracting every edge of $G[V_S(v)]$.
\end{definition}
If $G$ belongs to a minor-closed family $\mathcal K$ then so does $G_{\vor(S)}$.

\subsection{Euclidean space $r$-division}
\label{sec:Euclidean-r-division}
We define analogous notions for the case of Euclidean spaces of fixed dimension $d$.
Consider a set of points $C$ in $\R^d$.  For a set $Z$ of points in
$\R^d$ and a bipartition $C_1\cup C_2$ of $C$, we say that $Z$
\emph{separates} $C_1$ and $C_2$ if, in the Voronoi
diagram of $C\cup Z$, the boundaries of cells of points in $C_1$ are
disjoint from the boundaries of cells of points in $C_2$.

\begin{definition}\label{def:rdivision-euclid}
Let $c_1$ and $c_2$ be constants.  Let $C$ be a set of points in $\R^d$.
For an integer $r>1$, a  {\em
weak $r$-division} of $C$ (with respect to $c_1,c_2$) is a set of {\em boundary points} 
$Z \subset \R^d$ together with a collection of subsets $\mathcal R$ of
$C \cup Z$ called {\em regions}, with the following properties.
\begin{enumerate}
\item ${\mathcal R}\setminus Z$ is a partition of $C$.
\item The number of regions is at most $c_1 |C|/r$.\label{def:rdivisionb-euclid}
\item Each region contains at most $r$ points of $C \cup Z$.\label{def:rdivisionc-euclid}
\item $\sum_{R \in \mathcal{R}} |R \cap Z| \le c_2 |C|/r^{1/d}$.  \label{def:rdivisiond-euclid}
\end{enumerate}
Moreover, for any region $R_i$, $R_i \cap Z$ is a Voronoi separator for $R_i - Z$ and $(C \cup Z) - R_i$.
\end{definition}


The following theorem is from \cite[Theorem~3.7]{BH-P14}.
\begin{theorem} \cite[Theorem~3.7]{BH-P14}
  \label{thm:euclid-sep}
  Let $P$ be a set of $n$ points in $\R^d$. One can compute, in expected linear time, a sphere $S$,
  and a set $Z \subseteq S$, such that
  \begin{itemize}
  \item $|Z| \leq c n^{1-1/d}$,
  \item There are most $\sigma n$ points of $P$ in the sphere $S$ and at
    most $\sigma n$ points of $P$ not in $S$, and
  \item $Z$ is a Voronoi separator of the points of $P$ inside $S$ from the points of $P$ outside $S$.
  \end{itemize}
  Here $c$ and $\sigma <1$ are constants that depends only on the dimension $d$.
\end{theorem}

From that theorem we can easily derive the following (see
Section~\ref{sec:Euclidean-r-division-proof} for the proof):
\begin{theorem}
  \label{thm:r-division-euclid}
  Let $r$ be a positive integer and $d$ be fixed. Then there exist $c_1,c_2$ such that
  every set of points $C \subset \R^d$ has a weak $r$-division with respect
  to $c_1, c_2$.
\end{theorem}

\subsection{Properties of the $r$-Divisions}
We present the properties of the $r$-divisions that we will be using for the analysis of the solution
output by the local search algorithm.
\begin{lem}
  \label{lem:sep-graphs}
  Let $G=(V,E)$ be a graph excluding a fixed minor $H$ and $\calF \subseteq V$.
  Let $H_i$ be a region of the $r$-division of $G_{\vor(\calF)}$.
  Suppose $c$ and $v$ are vertices of $G$
  such that one of the vertices in $\set{\hat c, \hat v}$ is a
  vertex of $H_i$ and the other is {\em not} an internal vertex of $H_i$.
  Then there exists a vertex $x\in \calF$ such that $\hat x$ is a
  boundary vertex of the region $H_i$ and
  $\dist(c,x) \leq \dist(c,v)$.
\end{lem}
\ifshort
\else
\begin{proof}
  Let $p$ be a shortest $c$-to-$v$ path in $G$.  By the conditions on
  $\hat c$ and $\hat v$, there is
  some vertex $w$ of $p$ such that $\hat w$ is a boundary vertex of
  $H_i$.  Let $x$ be the center of the Voronoi cell whose contraction
  yields $\hat w$.  By definition of Voronoi cell,
  $\dist(w,x) \leq \dist(w, v)$.  Therefore replacing the $w$-to-$v$
  subpath of $p$ with the shortest $w$-to-$x$ path yields a path no
  longer than $p$.
\end{proof}
\fi
We obtain the analogous lemma for the Euclidean case, whose proof
follows directly from the definition of $r$-division 
(i.e.: the fact that $Z$ is a Voronoi separator).
\begin{lem}
  \label{lem:sep-euclid}
  Let $C$ be a set of points in $\R^d$ and $Z$ be an $r$-division of $C$.
  For any two different regions $R_1,R_2$, for any points $c \in R_1,v \in R_2$ 
  there exists a boundary vertex $x \in Z \cap R_1$ such that 
  $\dist(c,x) \leq \dist(c,v)$.
\end{lem}

\section{Facility Location in minor-closed graphs: Proof of Theorem~\ref{thm:FC}}
As a warm-up, we analyze Local Search for Uniform Facility Location
(Algorithm~\ref{algo:FacilityLocation}) applied to the metric
completion of an edge-weighted graph $G$ belonging to a nontrivial
minor-closed family $\mathcal K$.
The proof of the $k$-median and $k$-means results (for both Euclidean and 
minor-closed metrics), involve the use of Theorem \ref{thm:deletion}
and a more complex analysis.

Throughout this section we consider a solution $\calL$ output by
Algorithm \ref{algo:FacilityLocation} (the ``local'' solution) and
a globally optimal solution $\calG$ of value $\opt$. 
Let $\calF=\calL\cup\calG$. 
Let $r=1/\eps^2$.  Consider the graph $G_{\vor(\calF)}$ defined in
Definition~\ref{def:graph-Voronoi}, and recall that each vertex of $G$
maps to a vertex $\hat v$ in the contracted graph $G_{\vor(\calF)}$.

Since $G$ belongs to $\mathcal K$ and $G_{\vor(\calF)}$ is obtained
from $G$ by contraction, it too belongs to $\mathcal K$ and hence it
has an $r$-division.  Let $H_1,
\ldots H_{\kappa}$ be the regions of this $r$-division.   For $i=1,
\ldots, \kappa$, define $V_i$ and $B_i$ as follows:
\begin{eqnarray*}
V_i & =& \set{v\in \calF \ :\ \hat v \text{ is a vertex of } H_i}\\
B_i & =& \set{v\in \calF \ :\ \hat v \text{ is a boundary vertex of } H_i}
\end{eqnarray*}
That is, $V_i$ is the set of vertices in the union of the local
solution and the global solution
 that map via
contraction to vertices of the region $H_i$, and $B_i$ is the set of
vertices in the union that map to {\em boundary} vertices of $H_i$.

Let $\calG'= \calG \cup \bigcup_{i=1}^\kappa B_i$.

Fix a region $H_i$ of the $r$-division of $G_{\vor(\calF)}$. We define
$\calL_i = \calL \cap V_i$ and $\calG'_i = \calG' \cap V_i$.
  We consider the mixed solution $\calM^i$ defined as follows:
  $$\calM^i = (\calL \setminus \calL_i) \cup \calG'_i.$$
  \begin{lem}\label{lem:localneighbFL}
  $|\calM^i \setminus \local| + |\local \setminus \calM^i| \leq 1/\eps^{2}$.
\end{lem}
\ifshort
\else
\begin{proof}
  To obtain $\calM^i$ from $\calL$, one can
remove the vertices in $\calL\cap V_i$ that are not in
  $\cal G'$, and
add the vertices in ${\cal G'}\cap V_i$ that are not in
  $\calL$.
Thus the size of the symmetric difference is at most $|(\calL \cup
\calG') \cap V_i|$.  Since the vertices of $\calL\cup \calG'$ are
centers of Voronoi cells, these vertices all map to different vertices
in the contracted graph $G_{\vor(\calF)}$.  Therefore $|(\calL \cup
\calG') \cap V_i|$ is at most the number of vertices in region $H_i$,
which is at most $r=1/\eps^2$.
\end{proof}
\fi


\begin{figure} \centering
\includegraphics[scale=.7]{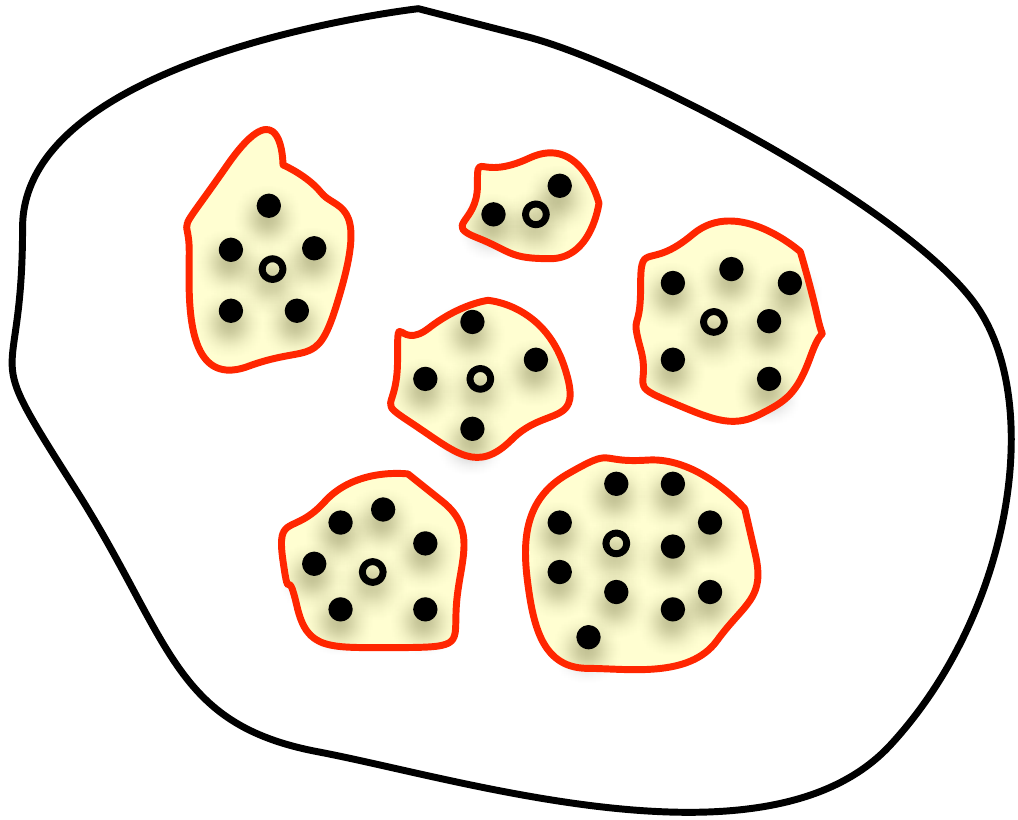}
\caption{The diagram shows a region of the weak $r$-division.  The
  blobs represent vertices of the region.  Each blob is obtained by
  coalescing a set of vertices of the input graph.  These vertices are
indicated by circles.  The unfilled circled represent the centers of
the Voronoi cells.}
\end{figure}

\begin{lem} \label{lem:client} 
Let $c$ be a vertex of $G$ and $H_i$ a region. Then:
$$m_c^i-l_c\leq\begin{cases}
g_c-\ell_c& \mbox{if $\hat c$ is an internal vertex of $H_i$}\\
0 & \mbox{otherwise.} \\
\end{cases}$$
\end{lem}
\ifshort
\else
\begin{proof} 
First suppose $\hat c$ is an internal vertex
of $H_i$, and let $v$ be the facility in $\calM^i$ closest to $c$.  If
$v$ is in $V_i$ then it is in $\calG_i$, so $m_c^i = g'_i$.  Suppose
$v$ is not in $V_i$.  Then by
Lemma~\ref{lem:sep-graphs} there is a vertex $x\in \calF$ such that $\hat x$
is a boundary vertex of $H_i$ and $\dist(c,x)\leq \dist(c,v)$.  As
before, $x$ is in $\calG'_i$ so $m_c^i \leq g'_c$.  Since $g'_c \leq g_c$, 
this proves the claimed upper bound. 

Now, suppose $\hat c$ is not an
internal vertex of $H_i$ and let $v$ be the facility in $\calL$ closest
to $c$.   If $v$ is not in $V_i$ then it is in the mixed solution
$\calM^i$, so $m_c^i=\ell_c$.  Suppose $v$ is in $V_i$.  Then by
Lemma~\ref{lem:sep-graphs} there is a vertex $x\in \calF$ such that $\hat x$
is a boundary vertex of $H_i$ and $\dist(c,x)\leq \dist(c,v)$.  Since
$x$ is in $\calF$ and $\hat x$ is a boundary vertex of $H_i$, we know
$x$ is in $\calG'\cap V_i$, which is $\calG'_i$.  Therefore $x$ is in
$\calM^i$.  Since $\dist(c,x)\leq \dist(c,v)$, we obtain $m_c^i\leq
\ell_c$, which proves the claimed upper bound. 
\end{proof}
\fi

\begin{lem} \label{lem:size-of-augmented-global-solution}
$$\sum_{i=1}^\kappa |\calG'_i| \leq |\calG| + c_2\eps
(|\cal G|+|\calL|)$$
\end{lem}
\ifshort
\else
\begin{proof} %
Let $v$ be a vertex of $\calG'$.  For $i=1, \ldots, \kappa$, if $\hat v$ is an internal vertex of
the region $H_i$ then $v$ contributes only one towards the left-hand
side.  If $\hat v$ is a boundary vertex of $H_i$ then $v\in B_i$.
Therefore 
$$\sum_{i=1}^\kappa |\calG'\cap V_i| \leq  |\calG| + \sum_{i=1}^\kappa
|B_i|.$$
To finish the proof, we bound the sum in the right-hand side.
  Each vertex in $\calF$ is the center of one Voronoi
  cell, so $G_{\vor(\calF)}$ has $|\calF|$ vertices.  For each region
  $H_i$, there is one vertex in $B_i$ that corresponds to each
  boundary vertex of $H_i$, so $\sum_{i=1}^\kappa
|B_i|$ is the sum over all regions of the number of boundary vertices
of that region, which, 
by
  Property~\ref{def:rdivisiond} of $r$-divisions, is at most $c_2
  |\calF|/r^{1/2}$, which, by choice of $r$, is at most $c_2 \eps
  |{\cal F}|$, which in turn is at most $c_2 \eps(|{\cal G}|+|\calL|)$.

\end{proof}
\fi

\begin{proof}(Proof of Theorem~\ref{thm:FC})
Lemma~\ref{lem:localneighbFL} and the stopping condition of
 Algorithm~\ref{algo:FacilityLocation} imply the following: 

  \begin{equation}
    \label{eq:local_hypo}
-{1\over n}  \cost(\calL)\leq     \cost({\calM}^i) -\cost(\calL).
  \end{equation}

  We now decompose the right-hand side. For a client $c$, we denote by $\ell_c$, $g'_c$ and $m_c^i$
  the distance from the client $c$ to the closest facilities in $\calL$, $\calG'$ and $\calM^i$ respectively.   This gives
\begin{equation} \label{eq:local_hypo-decomposed}
   \cost(\calM^i) -\cost(\calL)= (|\calG'_i| - |\calL_i| ) \cdot f 
  + \sum_{c} ( m_c^i - \ell_c  ).
\end{equation}
Using Lemma~\ref{lem:client} and summing over $c$ shows that 
\begin{equation}\label{eq:clientsum}
\sum_{c} ( m_c^i - \ell_c  ) \leq \sum_{c: \hat
  c\in \mathcal I(H_i)}  ( g_c - \ell_c  ).
  \end{equation}
Combining  Inequalities~(\ref{eq:local_hypo}),~(\ref{eq:local_hypo-decomposed}) and~(\ref{eq:clientsum}), we obtain
\begin{equation}
-{1\over n}  \cost(\calL) \leq (|\calG'_i| - |\calL_i| ) \cdot f 
  + \sum_{c: \hat c\in \mathcal I(H_i)}  ( g'_c - \ell_c  )
\end{equation}
We next sum this inequality over all $\kappa$ regions of the weak $r$-division and use Lemma~\ref{lem:size-of-augmented-global-solution}.
\begin{eqnarray*}
-{\kappa \over n}  \cost(\calL)
&\leq & 
 \left(\sum_{i=1}^\kappa|\calG'_i| - \sum_{i=1}^\kappa |\calL_i| \right) \cdot f + \sum_{i=1}^\kappa \sum_{c: \hat c\in \mathcal I(H_i)} (g_c-\ell_c)\\
& \leq & 
( |\calG| + \left(c_2\eps (|\calG|+|\calL|) -|\calL|\right)\cdot f + \sum_c (g_c-\ell_c)\\
& = & 
\left((1+ c_2\eps)|\calG| - (1-c_2\eps)|\calL|\right)\cdot f + \sum_c (g_c-\ell_c)\\
& \leq & 
(1+ c_2\eps) \cost(\calG) - (1-c_2\eps) \cost(\calL)
\end{eqnarray*}
Since $\kappa \leq c_1 |\calF|/r\leq c_1\eps^2 n$, we obtain
$$-c_1\eps^2 \cost(\calL) \leq (1+c_2 \eps) \cost(\calG) - (1-c_2\eps) \cost(\calL)$$
so
$$\cost(\calL) \leq (1-c_2\eps - c_1 \eps^2)^{-1} (1+ c_2 \eps)\cost(\calG)$$
This completes the proof of Theorem~\ref{thm:FC}.
\end{proof}

\section{Clusters in minor-closed graphs: Proof of Theorem \ref{thm:Clustering}}


\label{sec:AlgConseq}

We prove Theorem \ref{thm:Clustering}. The proof is similar for graphs and
for points lying in $R^d$.  It builds on the notions of isolation and
1-1 isolation introduced in Section~\ref{sec:isolation}.  

We consider a solution $\local$ output by Algorithm \ref{algo:Clustering} and an optimal solution $\globalS$.
Let $\bar \calF$ be the set of facilities of $\local$ and $\globalS$ that are not in a 1-1 isolated region
and let $\bar k = |\bar \calF|$.

We apply Theorem \ref{thm:deletion} to $\globalS$ and $\local$ in order to find a set $S_0 \subset \globalS$ 
such that $\cost(\globalS - S_0) \le (1+2^{3p+1} \eps) \cost(\globalS) +  2^{3p+1} \eps\,\cost(\local)$ and 
$|S_0|\ge \eps^3\bar k/6$.
Let $\globalS_1=\globalS \setminus S_0$.
We define a subgraph $G' = (V',E')$ of $G$ as follows: For each isolated
  region $(\local_0,f_0)$, for each client $c\in \Vloc(\local_0)\cap
  \Vglob(f_0)$, designate $c$ as a {\em good} client, and 
include in $E'$ the edges of a
  c-to-$L_0$ shortest path and a $c$-to-$f_0$ shortest path.
For every nonisolated facility $\ell\in \local$ and every nonisolated
facility $f\in \globalS$,  for every client $c \in \Vloc(\ell)\cap
\Vglob(f)$, also designate $c$ as a {\em good} client, and include in
$E'$ the edges of a  c-to-$\ell$ shortest path and a $c$-to-$f_0$ shortest path.
Let $\mathcal C_1$ be the set of clients designated as {\em good}.  The
remaining clients are considered {\em bad}.  

Let $\calF = \globalS_1 \cup \local$.
Let $R_1, R_2, \ldots$ be an $r$-division of $G'_{\vor( \calF)}$ where
$r=1/\eps^7$.  Define $\globalSs = \globalS_1 \cup \set{\text{boundary vertices of the $r$-division}}$.

\newcommand{\OmegaSubset}{\widehat{\mathcal G}}

The vertex sets of regions are of course not disjoint---a boundary vertex
is in multiple regions---but it is convenient to represent them by
disjoint sets.  We therefore define a ground set $\Omega= \set{(v, R)\ : v \text{ a vertex of }
  G'_{\vor(  \calF)}, R \text{ a region containing } v}$, and, for
each region $R$, we define $\widehat{R}=\set{(v,R)\ :\ v \text{ a
    vertex of } R}$. Now $\widehat{R_1}, \widehat{R_2}, \ldots$ form a
partition of $\Omega$.  To allow us to go from an element of $\Omega$
back to a vertex, if $x=(v, R)$ we define $\widecheck x= v$.
Finally, define $\OmegaSubset=\set{(v,R)\in \Omega :\ v\in \globalSs}$.





Let $\bar \calF$ be the set of facilities of $local$ and $\globalS$
that are not in 1-1 isolated regions.

\begin{lem}
  \label{lem:boundSs}
  $|\OmegaSubset|\leq |\globalS_1| + c_2  \eps^{3.5} |\bar \calF|$,
  where $c_2$ is the constant in the definition of $r$-division.
\end{lem}
\ifshort
\else
\begin{proof}
  Consider the $r$-division.  Each 1-1 isolated region results in a connected component of size 2 in $G'_{\vor( \calF)}$ and so no
  boundary vertices arise from such connected components.
By the definition of $r$-division, the sum over regions of boundary vertices is at most $c_2 \cdot |\bar n_0| / r^{1/2}$, where $\bar n_0$  is the total number
  of elements of $\globalS_1$ and $\local$ that are not in 1-1
  isolated regions.
Since $r = 1/\eps^7$, 
  we have that $|\OmegaSubset| \le |\globalS_1| + c_2 \cdot \eps^4 |\bar \calF|$.
\end{proof}
\fi

\begin{lem}
  \label{lem:globstruct-cardinality}
 $|\OmegaSubset| \le k$. \label{prop:feasible}
\end{lem}
\ifshort
\else
\begin{proof}
    
  By Theorem \ref{thm:deletion}, we have that $|\globalS_1| \le k - \eps^3 \bar k/12$.
  By Lemma \ref{lem:boundSs} we thus have
  $$|\OmegaSubset| \le |\globalS_1| + c_2 \eps^{3.5} \bar k \le  k - \eps^3 \bar k/12 + c_2 \eps^{3.5} \bar k \le k,$$
  for $\eps$ small enough.  
\end{proof}
\fi 

Throughout the rest of the proof, we will bound the cost of $\local$ by the cost of $\globalSs$.
We now slightly abuse notations in the following way : each facility $\ell$ of $\local$ that belongs to an isolated region
and that is a boundary vertex is now in $\globalSs$. We say that this facility is isolated.

The following lemma first appears in~\cite{CM15}.

\begin{lem}[Balanced Partitioning]
  \label{lem:kmed_partition}
  Let $\mathcal{S} = \{S_1,...,S_p\}$ and $\set{A,B}$  be partitions of some ground
  set.  Suppose $|A|\geq |B|$ and, for $=1,\ldots, p$, $1/(2\eps^2)
  \leq |S_i| \leq 1/\eps^2$.  

  There exists a partition that is a coarsening of $\mathcal S$
  satisfying the two following properties. For any part $C$ of the
  coarser partition,
  \begin{itemize}
  \item \textbf{Small Cardinality}:  $C$ is the union of
    $\mathcal{O}(1/\eps^5)$ parts of $\mathcal S$.
  \item \textbf{Balanced}: $|C\cap A| \geq |C \cap B|$.
  \end{itemize}
\end{lem}

We now apply Lemma \ref{lem:kmed_partition} to the partition
$\widehat{R_1}, \widehat{R_2}, \ldots$ with $A=\set{(v,R)\in \Omega\ :\
  v\in \local}$ and $B=\OmegaSubset$. We refer to the parts of the
resulting coarse partition as {\em super-regions}.  Each super-region
$\calR$ naturally corresponds to a subgraph of $G'_{\vor( \calF)}$,
the subgraph induced by $\set{v\ :\ (v,R)\in \calR}$,
and we sometimes use $\calR$ to refer to this subgraph.

For a super-region $\calR$, let $\local(\calR)$ (resp. $\globalSs(\calR)$) be the set of facilities
of $\local$ (resp. $\globalSs$) in the super-region $\calR$, i.e.: the set $\set{\ell \mid \ell \in \local \text{ and } (\ell,R) \in \Omega}$
(resp. $\set{f \mid f \in \globalSs \text{ and } (f,R) \in \Omega}$).
We consider the mixed solution $$\calM_{\calR} =( \local \setminus \local(\calR)) \cup \globalSs(\calR).$$

\begin{lem}\label{lem:localneighb}
  $|\calM_{\mathcal{R}} \setminus \local| + |\local \setminus \calM_{\mathcal{R}}| = O(1/\eps^{12})$ and $|\calM_{\mathcal{R}}| \le k$.
\end{lem}
\ifshort
\else
\begin{proof}
  Each region of the $r$-division contains at most $c_1 /\eps^{7}$
  facilities where $c_1$ is the constant in the definition of
  $r$-divisions. By Lemma \ref{lem:kmed_partition}, each super-region
  is the union of $\mathcal{O}(1/\eps^5)$ regions
\end{proof}
\fi

We now define $g_c$ to be the cost of client $c$ in solution $\globalSs$ and $l_c$ to be the cost
of client $c$ in solution $\local$.
For any client $c \in V_{\globalS}(f_0) \setminus V_{\local}(\local_0)$ for some isolated region $(f_0,\local_0)$, 
define $\globtoloc(c)$ as the cost of assigning $c$ to the facility of $\local_0$ that is the closest to $f_0$.
We let $\eps_1$ be a positive constant that will be chosen later.

\begin{lem}
  \label{lem:globtoloc}
  Consider an isolated region $(f_0, \local_0)$.
  $$\sum\limits_{c \in V_{\globalS}(f_0) \setminus V_{\local}(\local_0)} \globtoloc(c) \le \costPeps 
  \sum\limits_{c \in V_{\globalS}(f_0) \setminus V_{\local}(\local_0)} g_c
  +  \frac{2^p \costPeps \costPinv \eps}{1-\eps} \sum\limits_{c \in V_{\globalS}(f_0)} (g_c + l_c),$$

\end{lem}
\ifshort
\else
\begin{proof}
  Consider a client $c \in \Vglob(f_0) \setminus \Vloc(\local_0)$, and let $\ell$ denote the facility of $\local$ that is the closest
  to $f_0$. By Lemma \ref{lem:TI},
  $\dist(c,\ell)^p \leq \costPeps(\dist(c,f_0)^p + \costPinv \dist(\ell,f_0)^p) = \costPeps(g_c+ \costPinv \dist(\ell,f_0)^p)$. 
  Summing over $c\in \Vglob(f_0) \setminus \Vloc(\local_0)$,
  $$    \sum\limits_{c \in V_{\globalS}(f_0) \setminus V_{\local}(\local_0)} \loctoglob(c) \le 
  \costPeps\sum\limits_{c \in V_{\globalS}(f_0) \setminus V_{\local}(\local_0)} g_c + \costPeps \costPinv | \Vglob(f_0) \setminus \Vloc(\local_0)| 
  \dist(\ell,f_0)^p.
  $$
  To upper bound $ \dist(\ell,f_0)^p$, we use an averaging argument. 
  For each client $c' \in V_{\local}(\ell) \cap V_{\globalS}(f_0)$, let $\local(c)$ be the facility of $\local_0$ 
  that serves it in $\local$. By Lemma \ref{lem:TI} we have 
  $\dist(\ell,f_0)^p \leq 2^p (\dist(\ell,c')^p+\dist(c',f_0)^p) \leq 2^p(\dist(\local(c),c')^p + \dist(c',f_0)^p) =2^p(l_{c'}+g_{c'})$, thus 
  $$\dist(\ell,f_0)^p \le \frac{2^p}{|V_{\local}(\ell) \cap V_{\globalS}(f_0)| }
  \sum\limits_{c \in V_{\local}(\ell) \cap V_{\globalS}(f_0)} (l_c + g_c).
  $$
  Substituting, we have that $\sum\limits_{c \in V_{\globalS}(f_0) \setminus V_{\local}(\local_0)} \loctoglob(c)$ is at most
  $$
  \costPeps \sum\limits_{c \in V_{\globalS}(f_0) \setminus V_{\local}(\local_0)} g_c + 
  2^p \costPeps \costPinv \frac{|  V_{\globalS}(f_0) \setminus V_{\local}(\local_0)|}{ |V_{\local}(\ell) \cap V_{\globalS}(f_0)|}
  \sum\limits_{c \in V_{\local}(\ell) \cap V_{\globalS}(f_0)} (l_c + g_c).
  $$
  By definition of isolated regions,  $V_{\globalS}(f_0) \setminus V_{\local}(\local_0) \le \eps |\Vglob(f_0)|$ and  
  $|V_{\globalS}(f_0) \setminus V_{\local}(\local_0)| \ge (1-\eps) |\Vglob(f_0)|$, so the ratio is at most $\eps / (1-\eps)$. 
  Summing over $\ell\in \local_0$ proves the Lemma.
  
\end{proof}
\fi

Similarly, for any client $c \in V_{\local}(\local_0)  \setminus V_{\globalS}(f_0)$ for some isolated region $(f_0,\local_0)$, 
define $\loctoglob$ as the cost of assigning $c$ to $f_0$.

\begin{lem}
  \label{lem:loctoglob}
  Consider an isolated region $(f_0, \local_0)$.
  $$\sum\limits_{c \in V_{\local}(\local_0) \setminus V_{\globalS}(f_0)} \loctoglob(c) \le  \costPeps
  \sum\limits_{c \in V_{\local}(\local_0) \setminus V_{\globalS}(f_0)} l_c
  + \frac{2^p \costPeps \costPinv \eps}{1-\eps}  \sum\limits_{c \in V_{\globalS}(f_0)} (g_c + l_c).$$
\end{lem}
\ifshort
 \else
\begin{proof}
  Consider a client $c \in V_{\local}(\local_0) \setminus V_{\globalS}(f_0)$, and let $\ell$ denote the facility serving it in $\local$. 
  By Lemma \ref{lem:TI},
  $\dist(c,f_0)^p \leq \costPeps (\dist(c,\ell)^p + \costPinv \dist(\ell,f_0)^p)=\costPeps(\ell_c+\costPinv \dist(\ell,f_0)^p)$. 
  Summing over $c\in V_{\local}(\ell) \setminus V_{\globalS}(f_0)$,
  $$\sum\limits_{c \in V_{\local}(\ell) \setminus V_{\globalS}(f_0)} \loctoglob(c) \le \costPeps \left(
  \sum\limits_{c \in V_{\local}(\ell) \setminus V_{\globalS}(f_0)} l_c + \costPinv | V_{\local}(\ell) \setminus V_{\globalS}(f_0)| \dist(\ell,f_0)^p\right).
  $$
  To upper bound $ \dist(\ell,f_0)$, we use an averaging argument. For each client $c' \in V_{\local}(\ell) \cap V_{\globalS}(f_0)$, 
  by Lemma \ref{lem:TI} we have 
  $\dist(\ell,f_0)^p \leq 2^p (\dist(\ell,c')^p+\dist(c',f_0)^p)=2^p(l_{c'}+g_{c'})$, thus 
  $$ \dist(\ell,f_0)^p \le \frac{2^p}{|V_{\local}(\ell) \cap V_{\globalS}(f_0)| }
  \sum\limits_{c \in V_{\local}(\ell) \cap V_{\globalS}(f_0)} (l_c + g_c).
  $$
  Substituting,
  $$
  \sum\limits_{c \in V_{\local}(\ell) \setminus V_{\globalS}(f_0)} \loctoglob(c) \le 
  \costPeps\left(  \sum\limits_{c \in V_{\local}(\ell) \setminus V_{\globalS}(f_0)} l_c + 
  \frac{2^p\costPinv| V_{\local}(\ell) \setminus V_{\globalS}(f_0)|}{ |V_{\local}(\ell) \cap V_{\globalS}(f_0)|}
  \sum\limits_{c \in V_{\local}(\ell) \cap V_{\globalS}(f_0)} (l_c + g_c)\right).
  $$
  By definition of isolated regions,  $| V_{\local}(\ell) \setminus V_{\globalS}(f_0)| \le \eps |V_{\local}(\ell)|$ and  $| V_{\local}(\ell) \cap V_{\globalS}(f_0)| \ge (1-\eps) |V_{\local}(\ell)|$, so the ratio is at most $\eps / (1-\eps)$. Summing over $\ell\in \local_0$ proves the Lemma.
\end{proof}
\fi

\begin{lem}
  \label{lem:presence}
  Consider an isolated region $(f,\local_0)$. 
  Let $\ell$ be a facility of $\local_0$.
  For any super-region $\calR$, $\calM_{\calR}$ contains $f$ or a facility that is at distance at most $\dist(\ell,f)$ from $f$.
\end{lem}
\ifshort
\else
\begin{proof}
  Since $\ell$ and $f$ belong to the same isolated region $(f,\local_0)$ and $\ell \in \local_0$, they belong to the same connected
  component of $G'_{\vor}(\calF)$. 
  Now consider a super-region $\calR$ which does not contain $\ell$. Then $\ell \in \local(\calR)$. Thus, either $f \in \calR$
  or by Lemma \ref{lem:sep-graphs}, a boundary element $\ell' \in \calR$ of the $r$-division is on the path from $\ell$ to $f$
  and $\dist(\ell',f) \le \dist(\ell,f)$.  Thus, $\ell' \in \calM_{\calR}$, proving the lemma.
\end{proof}
\fi

For a client $c$ and a super-region $\calR$, we define $m_{\calR}(c)$
to be the cost of $c$ in the mixed solution $\calM_{\calR}$.
Moreover, for each client $c$, we consider the facilities $v$ and $w$ that serve this client in solution $\local$ and $\globalSs$ respectively.
We define $l(c)$ to be an arbitrary pair $(v,R) \in \Omega$ and $g^*(c)$ to be an arbitrary pair $(w,R) \in \Omega$.
We slightly abuse notation and say that $(v,R)$ is isolated if $v$ belongs to one of the isolated regions.

\begin{lem}\label{lem:costregion-good}
Let $c$ be a good client and $\calR$ a super-region. The value of $m_{\calR}(c)-l_c$ is less than or equal to:
$$\begin{cases}
g_c - l_c& \mbox{if $g^*(c) \in \calR$}\\
0 &\mbox{otherwise}\\
\end{cases}$$
\end{lem}
\begin{proof}
  Observe that if $g^*(c) \in \calR$, then $m_{\calR}(c) \le g_c$ and the first case holds.
  Now, for any super-region $\calR \not\ni l(c),g^*(c)$, $\calM_{\calR}$ contains the facility serving 
  client $c$ in local. Thus its cost is at most $l_c$ and the second case holds.
  Finally, assume that $\calR$ contains $l(c)$ and does not contain $g^*(c)$. If $\hat c$ belongs to 
  $\calR$, then by the separation property of the $r$-division (see Lemmas \ref{lem:sep-graphs}, \ref{lem:sep-euclid}), $g^*(c) \in \calR$
  and $m_{\calR}(c) \le g_c$. Otherwise, $\hat c \notin \calR$, and so, by the separation property 
  there must be a boundary vertex of $\calR$ that is closer to $c$ than the facility that serves
  it in $\local$. Therefore, we have $m_{\calR}(c) \le l_c$ and the second case holds.
\end{proof}

We now turn to the bad clients.

\begin{lem}\label{lem:costregion-bad}
Let $c$ be a bad client and $\calR$ a super-region. The value of $m_{\calR}(c)-l_c$ is less than or equal to:
$$\begin{cases}
g_c-l_c& \mbox{if $\ell(c) \in \calR$ and $g^*(c) \in \calR$}\\
\globtoloc(c)  - l_c & \mbox{if $\ell(c) \in \calR$ and $g^*(c) \notin \calR$ and $g^*(c)$ is isolated}\\
g_c - l_c &\mbox{if $\ell(c) \notin \calR$ and $g^*(c) \in \calR$ and $g^*(c)$  is not isolated}\\
\loctoglob(c) - l_c &  \mbox{if $\ell(c) \in \calR$ and $g^*(c) \notin \calR$ and $g^*(c)$ is not isolated}\\
0&\mbox{otherwise.}
\end{cases}$$
\end{lem}
 \ifshort
\else
\begin{proof}
  Observe that the super-regions form a partition of the $l(c)$ and $g^*(c)$. 
  Let $\calR(\ell(c))$ be the region that contains $\ell(c)$ and $\calR(g^*(c))$ be the region that contains $g^*(c)$.
  If $\calR(\ell(c)) = \calR(g^*(c))$ then, the facility serving $c$ in $\globalSs$ is in $\calM_{\calR(\ell(c))}$, 
  hence $m_{\calR(\ell(c))}(c) \le g_c$.
  Moreover for any other region $\calR' \neq \calR(\ell(c))$, we have $\ell(c) \notin \calR'$ and so the facility 
  serving $c$ in $\local$ is in $\calM_{\calR'}$. Therefore $m_{\calR'}(c) \le l_c$.

  Thus, we consider $c$ such that $\calR(\ell(c)) \neq \calR(g^*(c))$.
  Since $c$ is bad, we have that $\ell(c)$ or $g^*(c)$ is isolated.
  Consider the case where $g^*(c)$ is isolated. The cost of $c$ in solution 
  $\calM_{\calR(\ell(c))}$ is, by Lemma \ref{lem:presence}, at most $\globtoloc(c)$ satisfying the lemma.
  Now, for any other region $\calR' \neq \calR(g^*(c))$, again we have 
  $\ell(c) \notin \calR'$ and so the facilitiy serving $c$ in $\local$ is in $\calM_{\calR'}$.
  Therefore, $m_{\calR'}(c) \le l_c$.
  
  Therefore, we consider the case where $c$ is such that $\calR(\ell(c)) \neq \calR(g^*(c))$ and 
  such that $g^*(c)$ is not isolated. Since $c$ is bad, $\ell(c)$ is isolated.
  Thence, by Lemma \ref{lem:presence}, the cost in solution $\calM_{\calR(\ell(c))}$ is at most $\loctoglob(c)$,
  satisfying the Lemma.
  Moreover, in solution $\calR(g^*(c))$, the cost is at most $g_c$.
  Finally, for any other region $\calR' \neq \calR(\ell(c)),\calR(g^*(c))$, 
  $\ell(c) \notin \calR'$ and so the facilitiy serving $c$ in $\local$ is in $\calM_{\calR'}$.
  Therefore, $m_{\calR'}(c) \le l_c$, concluding the proof of the lemma.
\end{proof}

We now partition the clients into three sets, $\Lambda_1,\Lambda_2,\Lambda_3$.
Let $\Lambda_1$ be the set of clients such that there exists a super-region $\calR$ such that 
$\ell(c) \in \calR$ and $g^*(c) \notin \calR$ and $g^*(c)$ is not isolated.
Let $\Lambda_2$ be the set of clients such that there exists a super-region $\calR$ such that 
$\ell(c) \in \calR$ and $g^*(c) \notin \calR$ and $g^*(c)$ is isolated.
Finally let $\Lambda_3 $ be the remaining clients : $\Lambda_3 = \calC \setminus \Lambda_1 \setminus \Lambda_2$.
The following corollary follows directly from combining Lemmas \ref{lem:costregion-good} and \ref{lem:costregion-bad} 
and by observing that the super-regions form a partition of the $l(c)$ and $g^*(c)$, and by the definition 
of $\Lambda_1,\Lambda_2,\Lambda_3$.

\begin{cor}
  \label{cor:costclient}
  For any client $c$, we have that 
  $$ \sum_{\calR} (m_{\calR}(c) - l_c) \le 
  \begin{cases}
    \loctoglob + g_c - 2l_c &\mbox{if $c \in \Lambda_1$}\\
    \globtoloc - l_c &\mbox{if $c \in \Lambda_2$} \\
    g_c - l_c &\mbox{if $c \in \Lambda_3$}\\    
  \end{cases}
  $$
\end{cor}

We now turn to the proof of Theorem \ref{thm:Clustering}.
\begin{proof}[Proof of Theorem \ref{thm:Clustering}]
  By Lemma \ref{lem:localneighb}, for any super-region $\calR$ the solution $\calM_{\calR}$ 
  is in the local neighborhood of $\local$. By local optimality, we have 
  $$(1-\eps/n) \sum_c l_c \le \sum_c m_{\calR}(c).$$
  Hence,
  $$-\frac{\eps}{n} \cost(\local) \le  \sum_c (m_{\calR}(c) - l_c).$$
  Observe that the number of regions is at most $k \le n$. Thus, summing over all regions we have 
  $$-\eps \cost(\local) \le  \sum_{\calR} \sum_c (m_{\calR}(c) - l_c).$$
  Inverting summations and applying Corollary \ref{cor:costclient}, we obtain
  $$-\eps \cost(\local) \le  \sum_{c \in \Lambda_1} (\loctoglob(c) + g_c - 2l_c)  + \sum_{c \in \Lambda_2}(\globtoloc(c) - l_c)  + \sum_{c \in \Lambda_3} 
  (g_c - l_c).$$
  By definition of $\Lambda_1$ and since each client in $\Lambda_1$ is bad, applying Lemma \ref{lem:loctoglob} yields
  \begin{align*}
    -\eps \cost(\local) \le  \sum_{c \in \Lambda_1} (g_c + (\costPeps -2) l_c)  &+ \sum_{c \in \Lambda_2}(\globtoloc(c) - l_c)\\  &+ \sum_{c \in \Lambda_3} 
    (g_c -  l_c) + \frac{\costPeps \costPinv 2^p \eps}{1-\eps} (\cost(\local) + \cost(\globalSs)).
  \end{align*}
  Hence, for $\eps$ small enough with respect to $p$ and $\eps_1$, we have
  \begin{align*}
    -\eps \cost(\local) \le  \sum_{c \in \Lambda_1} (g_c - (1-\eps) l_c)  &+ \sum_{c \in \Lambda_2}(\globtoloc(c) - l_c)\\  &+ \sum_{c \in \Lambda_3} 
    (g_c -  l_c) + \frac{\costPeps \costPinv 2^p \eps}{1-\eps} (\cost(\local) + \cost(\globalSs)).
  \end{align*}
  Now, by definition of $\Lambda_2$ and since each client in $\Lambda_2$ is bad, applying Lemma \ref{lem:globtoloc} gives
  \begin{align*}
    -\eps \cost(\local) \le  \sum_{c \in \Lambda_1} (g_c - (1-\eps) l_c)  &+ \sum_{c \in \Lambda_2}(\costPeps g_c - l_c)\\  &+ \sum_{c \in \Lambda_3} 
                          (g_c - l_c) + \frac{2^{p+1} \costPeps \costPinv \eps}{1-\eps} (\cost(\local) + \cost(\globalSs))
  \end{align*}
  Thus,
  \begin{align*}
    -\eps \cost(\local) &\le \sum_c (\costPeps g_c - (1-\eps) l_c) +  
                          \frac{2^{p+1} \costPeps \costPinv \eps}{1-\eps} (\cost(\local) + \cost(\globalSs))\\ &\le
                          \costPeps (1+\frac{2^{p+1} \costPinv \eps}{1-\eps}) (\cost(\globalSs) - \cost(\local)),
  \end{align*}
  since $\Lambda_1,\Lambda_2,\Lambda_3$ is a partition of the clients.
  Therefore, assuming $\eps$ is small enough with respect to $p$ and
  $\eps_1$,
  there exists a constant $c_1$ such that 
  \begin{align*}
    (1-c_1\frac{2\eps}{1-\eps}-\eps)  \cost(\local) &\le (1+c_1\frac{2\eps}{1-\eps})  \cost(\globalSs)\\
    (1-c_1\frac{2\eps}{1-\eps}-\eps)  \cost(\local) &\le (1+c_1\frac{2\eps}{1-\eps})(1+\eps)  \cost(\globalS) + c_1\eps\cost(\local)
  \end{align*}
  Now, observe that $\cost(\globalSs) \le \cost(\globalS_1)$ since $\globalS_1 \subseteq \globalSs$. By Theorem
  \ref{thm:deletion}, $\cost(\globalS_1) \le (1+c_1\eps) \cost(\globalS) + c_1\eps\cost(\local)$.
  Combining concludes the proof of Theorem \ref{thm:Clustering}
\end{proof}

\section{Clusters in Euclidean space : Proof of Theorem \ref{thm:ClusteringEuclid}}
The proof is similar for $\R^d$. We explain how to modify the beginning of the proof
of the graph case, the rest of the proof applies directly.
We let $\calC_1$ denote the set of clients that do not belong to the symetric difference 
of $\Vloc(\local_0)$ and $\Vglob(f_0)$ of any isolated region $(\local_0,f_0)$.
We call them {\em good} clients; the others are {\em bad} clients. 
Again, we define a solution $\globalS_1$ by applying Theorem \ref{thm:deletion} to $\globalS$.
Let $\calF =  \local \cup \globalS_1$.
We now consider each isolated region $(\local_0,f_0)$, with $|\local_0| > 1/\eps^{7d}-1$, and proceed to an $r$-division 
of $\local_0 \cup \{f_0\}$ with $r = 1/\eps^{7d}$.
Moreover, for the remaining facilities $\calF$ of that are not in any isolated region, we proceed to
an $r$-division of those points with $r = 1/\eps^{7d}$.
We denote by $R_1,R_2\ldots$ the subset of all the regions defined by the above $r$-divisions.
Let $Z$ denote the set of boundary elements of all the $r$-divisions.
Define $\globalSs = \globalS_1 \cup Z$.

The point sets of regions are not disjoint since points of $Z$ appear in various regions. Thus,
we again define a ground set $\Omega= \set{(v, R)\ : v \text{ a point of } 
\calF,~R \text{ a region containing } v}$, and, for
each region $R$, we define $\widehat{R}=\set{(v,R)\ :\ v \text{ a
    point of } R}$. Now $\widehat{R_1}, \widehat{R_2}, \ldots$ form a
partition of $\Omega$.  To allow us to go from an element of $\Omega$
back to a point, if $x=(v, R)$ we define $\widecheck x= v$.
Finally, define $\OmegaSubset=\set{(v,R)\in \Omega :\ v\in \globalSs}$.

We now branch with the rest of the proof of Theorem \ref{thm:Clustering}, starting from Lemma \ref{lem:boundSs}.
\section{Reducing the number of clusters : Proof of Theorem \ref{thm:deletion}}

We recall the statement of Theorem \ref{thm:deletion}.


{
\renewcommand{\thetheorem}{\ref{thm:deletion}}
\begin{theorem}
  Let $\eps <1/2$ be a positive constant and $\local$ and $\globalS$ be two solutions for the $k$-clustering problem with exponent $p$.
  Let $\bar k$ denote the number of facilities $f$ of $\globalS$ that are not in a 1-1 isolated region. 
  There exists a set $S_0$ of facilities of $\globalS$ of size at least $\eps^3 \bar k/6$ that can be removed 
  from $\globalS$ at low cost:  $\cost(\globalS \setminus S_0) \le (1+2^{3p+1}\eps) \cost(\globalS) +  2^{3p+1}\eps\cost(\local)$.
\end{theorem}
\addtocounter{theorem}{-1}
}
Let $\eps <1/2$ be a positive constant and $\local$ and $\globalS$ be two solutions for the $k$-clustering problem with exponent $p$. 
Observe  that since $\eps < 1/2$, each facility of $\local$ belongs to at most one isolated region.
Let $\tilde \globalS$ denote the facilities of $\globalS$ that are not in an isolated region.
Theorem~\ref{thm:deletion} relies on the following lemma, whose proof we momentarily defer.
\begin{lem}
  \label{lem:redirect}
  There exists a function $\phi : \tilde \globalS \mapsto \globalS$ such that 
  reassigning all the clients of $V_{\globalS}(f)$ to $\phi(f)$ for every facility $f \in \tilde \globalS$
  increases the cost of $\globalS$ by at most $2^{3p+1} \eps^{-2} (\cost(\local) + \cost(\globalS))$.
\end{lem}
\begin{proof}[Proof of Theorem~\ref{thm:deletion}]
  Consider the abstract graph $H$ where the nodes are the elements of $\globalS$ and 
  there is a directed arc from $f$ to $\phi(f)$. More formally, $H = (\globalS, \{ \langle f,\phi(f) \rangle \mid f \in \tilde \globalS\})$.
  Notice that every node of $H$ has outdegree at most 1. Thus, there exists a coloring of the nodes of $H$ with three colors,
  such that all arcs are dichromatic.
  Let $S$ denote the color set with the largest number of nodes of $\tilde \globalS$. We have that $S$ contains at least $|\tilde \globalS|/3$ 
  nodes of $\tilde \globalS$.

  Arbitrarily partition $S$ into $1/\eps^3$ parts, each of cardinality at least $\eps^3 |\tilde \globalS| /3$. 
  By Lemma \ref{lem:redirect} and an averaging argument, there exists a part $S_0$ such that reassigning each facility $f \in S_0$ to $\phi(f)$ 
  increases the cost by at most 
  $$\frac{2^{3p+1}\eps^{-2}}{\eps^{-3}} (\cost(\local) + \cost(\globalS)) = 2^{3p+1}\eps(\cost(\local) + \cost(\globalS)).$$

  Since the arcs of $H$ are dichromatic, if $f \in S_0$ then $\phi(f) \notin S_0$.
  Consider the solution $\globalS \setminus S_0$. Client that belong to $V_{\globalS}(f)$ for 
  some $f \in S_0$ can be assigned in $\globalS \setminus S_0$ to a facility that is no farther than $\phi(f)$.
  Therefore, the cost of the solution $\globalS \setminus S_0$ is at most 
  $\cost(\globalS) + 2^{3p+1} \eps \cdot (\cost(\local) + \cost(\globalS))$.

  We now relate $|\tilde \globalS|$ to $\bar k$.
  Let $k(\local)_1$ be the number of facilities of $\local$ that belong to an isolated region that is not 1-1 isolated.
  Let $k(\globalS)_1$ be the number of facilities of $\globalS$ that belong to an isolated region that is not 1-1 isolated. 
  Finally, let $k(\globalS)_{2} = |\tilde \globalS|$.
  By definition, we have $k(\globalS)_1 + k(\globalS)_2 = \bar k \ge k(\local)_1 $.
  
  Now, observe that there are at least two facilities of $\local$ per isolated region that is not 1-1 isolated.
  Thus, $2 k(\globalS)_1 \le k(\local)_1$. Hence, $\bar k = k(\globalS)_1 + k(\globalS)_2 \le k(\local)_1/2 + k(\globalS)_2$.
  But for any $k(\globalS)_2 < k(\globalS)_1$, $k(\local)_1/2 + k(\globalS)_2 < k(\local)_1 \le \bar k$.
  Therefore, we must have $k(\globalS)_2 \ge k(\globalS)_1$, and so $k(\globalS)_2 \ge \bar k/2$.
  Thence $\eps |\tilde \globalS|/3 \ge \eps \bar k/6$ and the theorem follows.
\end{proof}
\ifshort
\else

We now define $g_c$ to be the cost of client $c$ in solution $\globalSs$ and $l_c$ to be the cost
of client $c$ in solution $\local$.

\begin{proof}[Proof of Lemma \ref{lem:redirect}]
For each facility $f \in \tilde \globalS$, we define $\phi(f) = \text{argmin}\{\dist(f,f') \mid f' \in \globalS \setminus \{f\}\}$.
Instead of analyzing the cost increase when reassigning clients of $\Vglob(f)$ to $\phi(f)$ we will analyze the cost increase of 
the following fractional assignment.
First for a facility $f \in \globalS$, 
we denote by $\hat \local(f)$  
the set  
  \begin{equation}
    \label{eq:loc2glob}
    \hat \local(f)=\{ \ell  \in \local \mid 1\leq |V_\globalS(f) \cap V_\local(\ell )| < (1-\eps)|V_\local(\ell )| \} .
  \end{equation}

  By definition of isolated regions (Definition~\ref{defn:isolated}), 
   for any $f \in \tilde \globalS$ we have
  \begin{equation}
    \label{eq:glob2loc}
    \sum\limits_{\ell \in \hat \local(f)} | \Vloc(\ell) \cap \Vglob(f) | > \eps|\Vglob(f)|.
  \end{equation}
  Thus, we partition the clients in $\Vglob(f)$ into parts indexed by $\ell \in \hat \local(f)$, in a such a way 
  that the part associated to $\ell$ has size at most $\eps^{-1} |\Vglob(f) \cap \Vloc(\ell)|$.
  For any $\ell \in \hat \local(f)$, the clients in the associated part are reassigned to the facility $\psi(\ell,f) 
  \in \globalS \setminus \{f\}$ that is the closest to $\ell$.

  We now bound the cost increase $\Delta$ induced by the reassignment. For each client $c \in \Vglob(f)$ assigned
  to a part associated to a facility $\ell$, the new cost for $c$
  is $\cost'_c = \dist(c,\psi(\ell,f))^p$. By the triangular inequality and Lemma \ref{lem:TI}, 
  $\cost'_c \le 2^p (\dist(c,f)^p + \dist(f,\psi(\ell,f))^p) = 2^p (g_c + \dist(f,\psi(\ell,f))^p)$.
  Summing over all clients, we have that the new cost is at most
  $$\sum_{c} 2^p g_c + \sum_{f \in \tilde \globalS} \sum_{\ell \in \hat \local(f)} \eps^{-1} |\Vglob(f) \cap \Vloc(\ell)|
  2^p \dist(f,\psi(\ell,f))^p.$$
  Let $\Delta = \sum_{f \in \tilde \globalS} \sum_{\ell \in \hat \local(f)} \eps^{-1} |\Vglob(f) \cap \Vloc(\ell)|
  2^p \dist(f,\psi(\ell,f))^p.$
  By Lemma \ref{lem:TI}, we have 
  $$\Delta \le \sum_{f \in \tilde \globalS} \sum_{\ell \in \hat \local(f)} \eps^{-1} |\Vglob(f) \cap \Vloc(\ell)|
  4^p (\dist(f,\ell)^p + \dist(\ell,\psi(\ell,f))^p).$$
  Inverting summations, 
  \begin{equation*}
    \Delta \le 4^p \eps^{-1}\left( \sum_{\ell \in \local} \sum_{f \in \tilde \globalS:\ell \in \hat \local(f)} |\Vglob(f) \cap \Vloc(\ell)| \dist(f,\ell)^p + 
      \sum_{\ell \in \local} \sum_{f \in \tilde \globalS:\ell \in \hat \local(f)} |\Vglob(f) \cap \Vloc(\ell)| \dist(\ell,\psi(\ell,f))^p\right).
  \end{equation*}

  Define $\Delta_1 = \sum_{\ell \in \local} \sum_{f \in \tilde \globalS:\ell \in \hat \local(f)} |\Vglob(f) \cap \Vloc(\ell)|
  \dist(f,\ell)^p $
  and $\Delta_2 = \sum_{\ell \in \local} \sum_{f \in \tilde \globalS:\ell \in \hat \local(f)} |\Vglob(f) \cap \Vloc(\ell)| \dist(\ell,\psi(\ell,f))^p.$

  We first bound $\Delta_1$. By Lemma \ref{lem:TI}, we have 
  that $\dist(f,\ell)^p \le 2^p(\dist(f,c)^p + \dist(\ell,c)^p) = 2^p(g_c + l_c)$ for any client $c \in \Vglob(f) \cap \Vloc(\ell)$.
  Therefore,
  \begin{align*}
    \Delta_1 &\le \eps^{-1} \sum_{\ell \in \local} \sum_{f \in \tilde \globalS~:~\ell \in \hat \local(f)} |\Vglob(f) \cap \Vloc(\ell)|
               \frac{2^p}{|\Vglob(f) \cap \Vloc(\ell)|} \sum_{c \in \Vglob(f) \cap \Vloc(\ell)} (g_c + l_c)\\
    &\le  2^p\eps^{-1} \sum_{\ell \in \local} \sum_{f \in \tilde \globalS~:~\ell \in \hat \local(f)}  \sum_{c \in \Vglob(f) \cap \Vloc(\ell)} (g_c + l_c) \le
               2^p\eps^{-1} (\cost(\globalS) + \cost(\local)).
  \end{align*}
  We now turn to bound the cost of $\Delta_2$.
  Let $f^{\ell}_{\min}$ be the facility of $\globalS$ that is the closest to $\ell$.
  Let
  \begin{align*}
    \Delta_3 &= \eps^{-1} \sum_{\ell \in \local}\sum_{f \neq f^{\ell}_{\min}:\ell \in \hat\local(f)} |\Vglob(f) \cap \Vloc(\ell)| \dist(\ell,\psi(\ell,f))^p\\
    \Delta_4 &= \eps^{-1} \sum_{\ell \in \local} |\Vglob(f^{\ell}_{\min}) \cap \Vloc(\ell)| \dist(\ell,\psi(\ell,f^{\ell}_{\min}))^p.
  \end{align*}
  For any client $c \in \Vglob(f) \cap \Vloc(\ell)$, by Lemma \ref{lem:TI}, $\dist(\ell,\psi(\ell,f))^p$, for $f \neq f^{\ell}_{\min}$ yields 
  $\dist(\ell,\psi(\ell,f))^p \le 2^p(\dist(\ell,c)^p + \dist(c,\psi(\ell,f))^p) \le 2^p(\dist(\ell,c)^p + \dist(c,f)^p) = 2^p(l_c + g_c)$.
  Thus, 
  $$\Delta_3 \le \eps^{-1} \sum_{\ell \in \local}\sum_{f \neq f^{\ell}_{\min}:\ell \in \hat\local(f)}
  \frac{2^p |\Vglob(f) \cap \Vloc(\ell)| }{|\Vglob(f) \cap \Vloc(\ell)|} \sum_{c \in \Vglob(f) \cap \Vloc(\ell)} (l_c + g_c) 
  \le 2^p\eps^{-1} (\cost(\globalS) + \cost(\local)).$$
  
  We conclude by analyzing $\Delta_4$. 
  Observe that if $\ell \notin \hat\local(f^{\ell}_{\min})$ then we are done : the clients in $\Vglob(f^{\ell}_{\min})$ are not reassigned
  through $\ell$. Thus we assume $\ell \in \hat\local(f^{\ell}_{\min})$.
  We now apply Lemma \ref{lem:TI} to $\dist(\ell,\psi(\ell,f^{\ell}_{\min}))^p$,
  for any client $c \in \Vloc(\ell) \setminus \Vglob(f^{\ell}_{\min})$ we have 
  $\dist(\ell,\psi(\ell,f^{\ell}_{\min}))^p \le 2^p(\dist(\ell, c)^p + \dist(c,\psi(\ell,f^{\ell}_{\min}))^p)
  \le 2^p(l_c + g_c)$, since $\psi(\ell,f^{\ell}_{\min})$ is the facility of $\globalS$ that is the second closest to $\ell$.
  Replacing we have, 
  \begin{align*}
    \Delta_4 &\le \eps^{-1} \sum_{\ell \in \local} \frac{2^p|\Vglob(f^{\ell}_{\min}) \cap \Vloc(\ell)| }
               {|\Vloc(\ell) \setminus \Vglob(f^{\ell}_{\min})|} \sum_{c \in \Vloc(\ell) \setminus \Vglob(f^{\ell}_{\min})} (l_c + g_c)
  \end{align*}
  Now, since $\ell \in \hat \local(f^{\ell}_{\min})$, we have that
  $|\Vglob(f^{\ell}_{\min}) \cap \Vloc(\ell)|/|\Vloc(\ell) \setminus \Vglob(f^{\ell}_{\min})|  \le (1-\eps)/{\eps}.$
  Therefore, 
  $$\Delta_4 \le 2^p(1-\eps)\eps^{-2} \sum_{c \in \Vloc(\ell) \setminus \Vglob(f^{\ell}_{\min})} (l_c + g_c).$$
  Putting $\Delta_1,\Delta_2,\Delta_3,\Delta_4$ together 
  we obtain that the total cost increase induced by the reassignment is at most
  $2^{3p+1}(\cost(\globalS) + \cost(\local))/\eps^2.$

\end{proof}
\fi

\section{Postponed proofs}
\subsection{Proof of existence of weak $r$-divisions in Euclidean space}
\label{sec:Euclidean-r-division-proof}
\begin{proof}
We describe a recursive procedure to construct the set $Z$ in the
definition of weak $r$-division of $C$.  Assuming that $|C|>r$, find a sphere $S$ and a set $Z_0$
satisfying Theorem~\ref{thm:euclid-sep}.  Let $Z_1$ be the result of
applying the procedure to the union of $Z_0$ with the set of points
inside $C$, and similarly obtain $Z_2$ from the set of points outside $C$.
Return $Z_0\cup Z_1 \cup Z_2$.

It is clear that the set $Z$ together with its induced partition $\calR$ of $C$ returned by the procedure satisfies all
the properties of a weak $r$-division except for Property~4, which
requires some calculation.  Let $b(n) = \sum_{R \in \mathcal{R}} |R \cap Z| $ when the procedure
is applied to a set $C$ of size at most $n$, where $n>(1-\sigma) r$.  If
$n \leq r$ then $b(n)=0$, and if $n > r$ then 
$$b(n) \leq cn^{1-1/d} + \max_{\alpha\in [1-\sigma,\sigma]} f(\alpha n + cn^{1-1/d}) + f((1-\alpha) n + cn^{1-1/d}).$$
We show by induction that $b(n) \leq \beta \frac{n}{r^{1/d}} - \gamma
n^{1-1/d}$ for suitable constants $\beta, \gamma>0$ to be determined. We postpone the basis of the induction until $\beta, \gamma$ are selected.

By the inductive hypothesis,
\begin{eqnarray*}
b(\alpha n + cn^{1-1/d}) & \leq & \beta  \frac{\alpha n}{r^{1/d}} + \beta
  \frac{cn^{1-1/d}}{r^{1/d}} - \gamma \alpha^{1-1/d} n^{1-1/d}\\
b((1-\alpha) n + cn^{1-1/d}) & \leq & \beta  \frac{(1-\alpha) n}{r^{1/d}} + \beta
  \frac{cn^{1-1/d}}{r^{1/d}} - \gamma (1-\alpha)^{1-1/d} n^{1-1/d}
\end{eqnarray*}
so
\begin{equation} \label{eq:rec}
b(n) \leq \left(c+ \frac{2c}{r^{1/d}}\right) n^{1-1/d} + \beta
\frac{n}{r^{1/d}} - \gamma \left[\alpha^{1-1/d} + (1-\alpha)^{1-1/d}\right]n^{1-1/d}
\end{equation}
The function $f(x) = x^{1-1/d} + (1-x)^{1-1/d}$ is strictly concave
for $x\in [0,1]$, as can be
seen by taking its second derivative.  For any $\alpha\in [1-\sigma,
\sigma]$, there exists a number $0 < \mu < 1$ such that
$\alpha = (1-\mu) (1-\sigma) + \mu \sigma$.  By concavity, therefore,
$f(\alpha) \geq (1-\mu)f(1-\sigma) + \mu f(\sigma)$.  Since a weighted
average is at least the minimum, $(1-\mu)f(1-\sigma) + \mu
f(\sigma) \geq \min \set{f(1-\sigma), f(\sigma)}$.  Write
$f(1-\sigma)=f(\sigma)=1+\delta$.  Since $f$ is strictly concave, $\delta>0$.
We choose $\gamma=(c+2c/r^{1/d})/\delta$, for then the first term in
Inequality~\ref{eq:rec} is bounded by $\gamma \delta n^{1-1/d}$, and
we obtain $b(n) \leq \beta \frac{n}{r^{1/d}} - \gamma
n^{1-1/d}$.

For the basis of the induction, suppose $n>(1-\sigma)r$.  Then

$$\beta \frac{n}{r^{1/d}} - \gamma
n^{1-1/d} = \left(\beta \frac{n^{1/d}}{r^{1/d}} -
  \gamma\right)n^{1-1/d} \geq \left(\beta
  \frac{(1-\sigma)^{1/d}r^{1/d}}{r^{1/d}} -\gamma\right) = \left(\beta
  (1-\sigma)^{1/d} -\gamma\right)$$
which is nonnegative for an appropriate choice of $\beta$ depending on
$\sigma$ and $\gamma$.
\end{proof}

\bibliographystyle{abbrv}
\bibliography{facilitylocationptas.bib}

\end{document}

